\documentclass[a4paper,10pt]{elsarticle}

\makeatletter
\def\ps@pprintTitle{%
  \def\@oddhead{}%
  \def\@evenhead{}%
  \def\@oddfoot{}%
  \def\@evenfoot{}%
}
\makeatother

\usepackage{amssymb}
\usepackage{amsmath}
\usepackage{xcolor}
\usepackage{soul}
\usepackage{amsthm}
\usepackage{hyperref}

\theoremstyle{plain}

\newtheorem{proposition}{Proposition}

\newtheorem{theorem}{Theorem}

\newtheorem{remark}{Remark}

\theoremstyle{definition}
\newtheorem{example}{Example}
\newtheorem{definition}{Definition}

\newcommand{\domain}{\mathcal{R} \mathfrak{(X)}}

\newcommand{\codomain}{\mathcal{R}(X)}

\usepackage{url}
\usepackage[T1]{fontenc}

\begin{document}

\begin{frontmatter}

\title{Independence Axioms in Social Ranking}

 \author[label1]{Takahiro Suzuki\textsuperscript{[0000-0002-3436-6831]}}
  \author[label2]{Michele Aleandri\textsuperscript{[0000-0002-5177-8176]}}
   \author[label3]{Stefano Moretti\textsuperscript{[0000-0003-3627-3257]}}
   
 \affiliation[label1]{organization={Department of Civil Engineering, 
 Graduate School of Engineering, The University of Tokyo},
             addressline={Hongo Campus, 7-3-1 Hongo, Bunkyo-ku}, 
            city={Tokyo},
            postcode={113-8656}, 
            country={Japan}
            suzuki-tkenmgt@g.ecc.u-tokyo.ac.jp}

 \affiliation[label2]{organization={LUISS University},
             addressline={Viale Romania, 32},
             city={Rome},
             postcode={00197},
             country={Italy} \\
             maleandri@luiss.it}
\affiliation[label3]{organization={LAMSADE, CNRS, Universit{\'e} Paris-Dauphine, Universit{\'e} PSL},
             city={Paris},
             postcode={75016},
             country={France} \\
stefano.moretti@lamsade.dauphine.fr}

\begin{abstract}
Independence from non-essential changes in input information is a widely recognized axiom in social choice theory. This independence reduces the cost of specifying and/or analyzing non-essential data. This study makes a comprehensive analysis of independence axioms in the context of social ranking solutions (SRSs). We consider seven independence axioms (two of which are new) and provide a novel characterization of the lexicographic excellence solution and plurality by substituting these independence axioms in the existing characterization of the intersection initial segment rule. The characterizations highlight the differences among the three SRSs in terms of independence. 
\end{abstract}

\begin{keyword}
Social ranking solution\sep Lexicographic excellence solution \sep Plurality \sep Independence axioms
\end{keyword}
\end{frontmatter}

\section{Introduction}\label{section:introduction}
Independence from irrelevant changes in input data is a well-known principle in social choice theory. A rule satisfying such independence enables us to focus only on the essential components of the input information, thereby lowering the expense of gathering or evaluating unnecessary information. The most well-known example of independence is the \textit{independence of the irrelevant alternatives} (IIA), introduced by \cite{Arrow1951} (see, e.g., \cite{Patty2019} for further justification of IIA). This study examines independence axioms in the model of social ranking solutions (SRSs). 

SRS is a function that derives a contribution ranking of individuals from a performance ranking of their coalitions \cite{Moretti2017}. The literature on SRSs includes numerous examples of independence axioms, such as \textit{independence of the irrelevant coalitions} (IIC) \cite{Moretti2017}, \textit{independence from the worst set} (IWS) and\textit{ independence from the best set} (IBS) \cite{Bernardi2019,Algaba2021,Aleandri2024}, \textit{slide independence} (SI) \cite{suzuki2025mill}, and \textit{independence of the non-unanimous improvements} (INUI) \cite{Suzuki2025}. These are considered instances of independence axioms because they all require that the outcome of an SRS must be independent of the "non-essential change" in the input information; their difference lies in what they assume non-essential. This study focuses on these five axioms along with two new independence axioms: \textit{strong SI} (SSI) and \textit{tops-only }(TO). 

We focus on three SRSs: \textit{lexicographic excellence solution} (lex-cel), \textit{plurality}, and \textit{intersection initial segment} (IIS). Table \ref{table:Characterizations in the literature} presents a literature review of axiomatic studies on these SRSs.

\begin{table}[h]
\centering
    \caption{Characterizations in the literature}
    \label{table:Characterizations in the literature}

\begin{tabular}{|c|c|c|}
\hline
lex-cel \cite{Bernardi2019} & NT, CA, MON, and \textbf{IWS} &  \\\hline
lex-cel \cite{Algaba2021} & NT, WCA, MON, and \textbf{IWS} & \\
\hline
lex-cel \cite{Suzuki2024}& NT, WCA, CON, and CD & *variable domain of coalitions\\
\hline
lex-cel \cite{Aleandri2024} & SD, SYM, CA, and \textbf{IWS} & \\
\hline
lex-cel \cite{Suzuki2024c} & LIN, SP, CA, and reward cond. & *linear domain\\
\hline
plurality \cite{suzuki2025mill}& STAG, DMON, and \textbf{SI} & \\
\hline
IIS \cite{Suzuki2025} & NT, WIVIP, \textbf{IWS, IBS, and INUI} & *opinion aggregator\\
\hline
\end{tabular}
\par
\vspace{5pt}
{\footnotesize The abbreviations in the table are as follows: NT: neutrality; CA: coalitional anonymity; MON: monotonicity; WCA: weak coalitional anonymity; CON: consistency; CD: complete dominance; SD: strict desirability; SYM: symmetry; LIN: linear; SP: sabotage-proof; STAG: strong top agreement; DMON: down-monotonicity. The axioms in bold font are instances of what we refer to as independence axioms.}
\end{table}

\begin{remark}
Axioms CA, NT and WCA in Table \ref{table:Characterizations in the literature} could also be interpreted as independence axioms. Indeed, these axioms state that a SRS  should remain invariant under permutations of some individuals (as in the NT axiom) or of certain coalitions (as in the CA and WCA axioms). While this invariance can be viewed as a form of independence, we emphasize that the independence axioms considered in this paper concern changes in the ranking of some coalitions and do not involve any notion of permutation (of individuals or of coalitions). Furthermore, the axioms CA, NT, and WCA do not play a direct role in the proofs of theorems presented in this study.  
\end{remark}

The main contribution of this study is to provide a new characterization of lex-cel (Theorem \ref{theorem:lex-cel}) and plurality (Theorem \ref{theorem:plurality}) by replacing independence axioms in the known characterization of IIS in \cite{Suzuki2025}. Table \ref{table:Overview of our results} provides an overview of these theorems. In the first row, NT is enclosed in the parentheses because it is not used in the characterization, although lex-cel satisfies NT.
\begin{table}[h]
\centering
    \caption{Overview of our results}
    \label{table:Overview of our results}

\begin{tabular}{|c|c|}
\hline
Lex-cel (Theorem \ref{theorem:lex-cel})
 & (NT,) WIVIP, \textbf{IWS, and SSI} \\
 \hline
 Plurality (Theorem \ref{theorem:plurality}) & NT, WIVIP, \textbf{TO, and SI} \\
 \hline
\end{tabular}

\par
\vspace{5pt}{
{\footnotesize The abbreviations in the table are as follows: SSI: strong slide independence; TO: tops-only, SI: slide independence.}}
\end{table}

Compared with the bottom row of Table \ref{table:Characterizations in the literature}
, our characterizations are derived by substituting only independence axioms. In this sense, our results reveal a hidden relationship between lex-cel, plurality, and IIS in terms of independence axioms. 

The remainder of this paper is organized as follows. Section \ref{section:model} describes the basic model, including the definition of the independence axioms. Section \ref{section:results} presents our results. Section \ref{section:conclusion} concludes the study.

\section{Model}
\label{section:model}
Let $A$ be a finite set. A binary relation $R\subseteq A\times A$ is called \textit{reflexive} if $aRa$ for all $a\in A$; \textit{complete} if, for any $a,b\in A$, either $aRb$ or $bRa$ holds; and transitive if, for any $a,b,c\in A$, when $aRb$ and $bRc$ hold, then $aRc$ also holds. A reflexive, complete, and transitive binary relation is called a \textit{weak order}. The set of all weak orders on $A$ is denoted by $\mathcal{R}(A)$. For a binary relation $R\subseteq A\times A$, its asymmetric part $P(R)$ and symmetric part $I(R)$ are defined as follows:
\begin{align*}
&P(R):=\{(a,b)\mid (a,b)\in R \text{ and } (b,a)\notin R\},\\
&I(R):=\{(a,b)\mid (a,b)\in R\text{ and }(b,a)\in R\}.
\end{align*}
For a binary relation $\succsim$, we usually denote $P(\succsim)$ as $\succ$ and $I(\succsim)$ as $\sim$. For a binary relation $R\subseteq A\times A$ and $B\subseteq A$, we define the restriction of $R$ into $B$, denoted as $R\mid_{B}$, as $R\mid_{B}:=R\cap(B\times B)$. For a positive integer $k$, we denote by $[k]:=\{m\in \mathbb{N}\mid 1\leq m\leq k \}$ the set of all positive integers equal to or less than $k$. 
For a binary relation $\succsim$ on $A$, we write $\succsim:\Sigma_{1}\succ\cdots\succ\Sigma_{l}$ if (i) for each $k\in[l]$, $\Sigma_{k}$ is an equivalence class, and (ii) $\Sigma_{1},\cdots,\Sigma_{l}$ are ordered according to $\succsim$ (i.e., for any $k,k'\in[l]$ with $k<k'$, we have $S\succ T$ for all $S\in\Sigma_{k}$ and $T\in\Sigma_{k'}$).

Let $X=\{1,2,\cdots,n\}$ be the set of individuals. Let $\mathfrak{X}:=2^X$ be the power set of $X$. 

\begin{definition}
\label{definition:SRS}
A \textit{SRS} is a function of $R:\domain\rightarrow \codomain$.
\end{definition}
To simplify the notation, for a SRS $R$ and $\succsim\in\domain$, we denote $R(\succsim)$, $P(R(\succsim))$, and $I(R(\succsim))$ as $R_\succsim$, $P_{\succsim}$, and $I_{\succsim}$, respectively. 

We assume that axiom $A$ is an \textit{independence axiom} if it requires the SRS $R$ to be robust against insignificant changes in its input. Typically, an independence axiom states that for any $\succsim,\succsim'\in\domain$ and any $x,y\in X$, if the difference between $\succsim$ and $\succsim'$ is not substantial with respect to $x$ and $y$, then $R_{\succsim}$ and $R_{\succsim'}$ are closely related in terms of $x$ and $y$ (e.g., $R_{\succsim}\mid_{\{x,y\}}=R_{\succsim'}\mid_{\{x,y\}}$ or $xP_{\succsim}y\Rightarrow xP_{\succsim'}y$). 

Seven independence axioms (IWS, IBS, TO, SI, SSI, INUI, and IIC) are defined as follows. 

\begin{definition}
\label{definition:IWS and IBS}
An SRS $R:\domain\rightarrow\codomain$ satisfies \textit{IWS} if, for any $\succsim,\succsim'\in\domain$ with $\succsim:\Sigma_{1}\succ\cdots\succ\Sigma_{l-1}\succ\Sigma_{l}$ $(l\geq 2)$ and $\succsim':\Sigma_{1}\succ'\cdots\succ'\Sigma_{l-1}\succ'\Gamma_{1}\succ'\cdots\succ'\Gamma_{l'}$, and for any $x,y\in X$, it holds that $xP_{\succsim}y\Rightarrow xP_{\succsim'} y$. Similarly, $R$ satisfies \textit{IBS} if for any $\succsim,\succsim'\in\domain$ with $\succsim:\Sigma_{1}\succ\Sigma_{2}\cdots\succ\Sigma_{l}$ $(l\geq 2)$ and $\succsim':\Gamma_{1}\succ'\cdots\Gamma_{l'}\succ'\Sigma_{2}\succ'\cdots\Sigma_{l}$, and for any $x,y\in X$, it holds that $xP_{\succsim}y\Rightarrow xP_{\succsim'} y$.
\end{definition}

IWS and IBS are proposed in \cite{Bernardi2019} to characterize lex-cel and its dual, respectively. IWS and IBS assume that decomposing the worst (best) equivalence class is not essential and that a strict relationship is preserved when such a decomposition is performed. 

\begin{definition}
\label{definition:TO}
An SRS $R:\domain\rightarrow\codomain$ satisfies \textit{TO} if, for any $\succsim,\succsim'\in\domain$ with $\succsim:\Sigma_{1}\succ\cdots\succ\Sigma_{l}$ and $\succsim:\Sigma'_1\succ'\cdots\succ'\Sigma'_{l'}$, if $\Sigma_{1}=\Sigma'_{1}$, then we obtain $R_{\succsim}=R_{\succsim'}$. 
\end{definition}

TO is a standard axiom in the study of plurality voting, which requires that social choice depends solely on voters’ top-ranked alternatives \cite{Yeh2008,Saitoh2022,Sekiguchi2012}. The TO is a direct application to the study of SRSs. Specifically, it stipulates that the outcome of the SRS depends only on the top equivalence class. Thus, it assumes that equivalence classes below the top are not essential for social rankings. 

\begin{remark}
\label{remark:independence from the ends} 
In a sense, IWS, IBS, and TO all relate to independence from the ``extremal elements''. IWS states that the SRS is independent of the worst elements (i.e., the worst equivalence class), IBS states that the SRS is independent of the best elements (i.e., the best equivalence class), and TO states that the SRS is independent of the classes that are worse than the best equivalence class.  
\end{remark}

To define the next axioms, we first need to introduce some further preliminary notation. For every element $x \in X$ and every family of sets $\Delta \subseteq \mathfrak{X}$, we denote as $\Delta[x]=\{S \in \Delta\mid x \in X\}$ the family of sets in $\Delta$ containing element $x$, and as $\bigcap \Delta=\bigcap_{S\in \Delta}S$ the intersection of sets contained in $\Delta$.

\begin{definition}
\label{definition:SI and SSI}
An SRS $R:\domain\rightarrow\codomain$ satisfies \textit{SI} (\textit{SSI}) if for any $\succsim,\succsim'\in\domain$ with $\succsim:\Sigma_{1}\succ\cdots\succ\Sigma_{l}$ and $\succsim':\Sigma'_{1}\succ'\cdots\succ'\Sigma'_{l}$, $k_1,k_2\in[l]$, $x,y\in X$,$\Delta\subsetneq\Sigma_{k_1}$ ($\Delta\subseteq \Sigma_{k_1}$) with $|\Delta[x]|=|\Delta[y]|$, $\Sigma'_{k_1}=\Sigma_{k_1}\setminus\Delta$, $\Sigma'_{k_2}=\Sigma_{k_2}\cup\Delta$, and $\Sigma'_{k}=\Sigma_{k}$ for all $k\neq k_1,k_2$, it holds that $R_\succsim\mid_{\{x,y\}}=R_{\succsim'}\mid_{\{x,y\}}$. 
\end{definition}

\begin{definition}
\label{definition:INUI}
An SRS $R:\domain\rightarrow\codomain$ satisfies \textit{INUI} if for any $\succsim\in\domain$ with $\succsim:\Sigma_{1}\succ\cdots\succ\Sigma_{l}$, $\Delta\subseteq\Sigma_{k}$ with $x,y\notin\bigcap \Delta$, it holds that $R_\succsim\mid_{\{x,y\}}=R_{\succsim'}\mid_{\{x,y\}}$, where $\succsim':\Sigma_{1}\succ'\cdots\succ'\Sigma_{k-1}\succ'\Delta\succ'\Sigma_{k}\setminus\Delta\succ'\cdots\succ'\Sigma_{l}$.  
\end{definition}

SI, SSI, and INUI all assume that shifting $\Delta$, a subset of some equivalence classes, is not essential to the social ranking between $x$ and $y$ as long as $\Delta$ satisfies some symmetry conditions with respect to $x$ and $y$. The primary difference among these axioms lies in the symmetry condition on $\Delta$. For SI and SSI, $\Delta$ must satisfy $|\Delta[x]|=|\Delta[y]|$ (i.e., the number of coalitions in $\Delta$ including $x$ is the same as that in $\Delta$ including $y$). For INUI, $x,y\notin \bigcap \Delta$ (i.e., none of $x$ and $y$ belong to the intersection of $\Delta$). Thus, SI and SSI focus on the number of coalitions, including the individuals, whereas INUI focuses on their intersections. 

Subsequently, SI requires that even if $\Delta$ shifts from an equivalence class $\Sigma_{k_{1}}(\supsetneq \Delta)$ to $\Sigma_{k_{2}}$, the social ranking between $x$ and $y$ remains the same. SSI further refers to the case where $\Sigma_{k_{1}}=\Delta$ (accordingly, the number of equivalence classes can decrease by the shift). INUI demands that the upward shift in $\Delta$ does not affect the social ranking between $x$ and $y$. 

SI is proposed in \cite{suzuki2025mill} to characterize plurality (the requirement of SI is close to the cancellation property in \cite{Barbera2023a,Young1974a}). Our SSI is a new axiom that strengthens SI. INUI is proposed in \cite{Suzuki2025} to characterize IIS. 

\begin{remark}
While IBS and IWS refer to only the best/worst equivalence class, SI refers to the slide of a family of subsets $\Delta$ between any two equivalence class. This motivates us to consider an apparently weaker axiom, \textit{Top slide independence} (\textit{Top-SI}), which demands the condition of SI only for $k_{1}=1$ (i.e., considering the slide from/into the top equivalence class). However, we can confirm that Top-SI is equivalent to SI. SI$\Rightarrow$Top-SI is straightforward. We briefly show the converse. Let $\succsim:\Sigma_{1}\succ\cdots\succ\Sigma_{l}$, and let $\hat{k}_{1},\hat{k}_{2}\in[l]$. We define $\succsim_{1}(:=\succsim),\succsim_{2},\succsim_{3}\in\domain$ as follows. 
\begin{align*}
&\succsim_{1}:\Sigma_{1}\succ_{1}\cdots\succ_{1}\Sigma_{\hat{k}_{1}}\succ_{1}\cdots\succ_{1}\Sigma_{\hat{k}_{2}}\succ_{1}\cdots\succ_{1}\Sigma_{l}\\
&\succsim_{2}:\Sigma_{1}\succ_{2}\cdots\succ_{2}\Sigma_{\hat{k}_{1}}\setminus\Gamma\succ_{2}\cdots\succ_{2}\Sigma_{\hat{k}_{2}}\cup\Gamma\succ_{2}\cdots\succ_{2}\Sigma_{l}\\
&\succsim_{3}:\Sigma_{1}\cup\Gamma\succ_{3}\cdots\succ_{3}\Sigma_{\hat{k}_{1}}\setminus\Gamma\succ_{3}\cdots\succ_{3}\Sigma_{\hat{k}_{2}}\succ_{3}\cdots\succ_{3}\Sigma_{l}.
\end{align*}

Then, we can apply Top-SI for the pair $\succsim_{3}$ and $\succsim_{2}$ (by $k_{1}=1 \mbox{ and } k_{2}=\hat{k}_{2}$) to obtain that $R_{\succsim_{3}}\mid_{\{1,2\}} =R_{\succsim_{2}}\mid_{\{1,2\}} $. Furthermore, we can apply 1-SI for the pair $\succsim_{3}$ and $\succsim_{1}$ (by $k_{1}=1$ and $k_{2}=\hat{k}_{1}$) to obtain that $R_{\succsim_{3}}\mid_{\{1,2\}} =R_{\succsim_{1}}\mid_{\{1,2\}} $. These imply that $R_{\succsim_{2}}\mid_{\{1,2\}} =R_{\succsim_{1}}\mid_{\{1,2\}} $, which means that the $R$ satisfies SI.
\end{remark}

\begin{definition}
\label{IIC}
An SRS $R:\domain\rightarrow\codomain$ satisfies \textit{IIC} if for any $\succsim,\succsim'\in\domain$, and for any $x,y\in X$, if $\left[S\cup\{x\}\succsim S\cup\{y\}\iff S\cup\{x\}\succsim' S\cup\{y\}\right]$ for all $S\subseteq X\setminus\{x,y\}$, then it holds that $xR_{\succsim}y\iff xR_{\succsim'}y$. 
\end{definition}

IIC requires that the social ranking between $x$ and $y$ be determined only by the ceteris paribus (CP) comparison $S\cup\{x\}$ and $S\cup\{y\}$ for each $S\subseteq X\setminus\{x,y\}$ \cite{Moretti2017,Suzuki2021} (see \cite{Suzuki2024a} for further discussion on CP comparison). In this study, IIC assumes that the transformation of $\succsim$ that preserves such CP comparison is not essential for determining the social ranking between $x$ and $y$. 

We have introduced seven independence axioms (IWS, IBS, TO, SI, SSI, INUI, and IIC). Next, we will introduce some complementary axioms. For a permutation $\pi:X\rightarrow X$ and $\succsim\in\domain$, we define $\succsim_\pi\in\domain$ as 
$$\succsim_\pi:=\{(\pi(S),\pi(T))\mid S\succsim T\}.$$

\begin{definition}
An SRS $R:\domain\rightarrow\codomain$ satisfies \textit{NT} if, for any permutation $\pi$ on $X$ and $x,y\in X$, we obtain $\pi(x)R_{\succsim_{\pi}}\pi(y)\iff xR_\succsim y$.
\end{definition}

NT is a standard axiom that assumes that individuals’ names are irrelevant in determining social rankings. It has also been referenced in literature on SRSs (\cite{Bernardi2019,Khani2019,Haret2019a}).

\begin{definition}
An SRS $R:\domain\rightarrow\codomain$ satisfies WIVIP if, for any $\succsim\in\domain$ with $\succsim:\Sigma_{1}\succ\Sigma_{2}$, and for any $x\in\bigcap \Sigma_{1}$ and $y\in(\bigcap\Sigma_{1})^c$, we have that $xP_{\succsim}y$.
\end{definition}

WIVIP is applied when we possess a dichotomous ranking $\succsim:\Sigma_{1}\succ\Sigma_{2}$, where every coalition is judged either as the better (elements of $\Sigma_{1}$) or the worse (elements of $\Sigma_{2}$). In such a case, WIVIP demands that individuals who belong to all the better coalitions are socially ranked higher than those individuals who do not. WIVIP is proposed in \cite{Suzuki2025}.

To formally define lex-cel $R^{L}$, plurality $R^{P}$, and IIS $R^{IIS}$, we introduce additional notation. For $\succsim\in\domain$ with $\succsim:\Sigma_{1}\succ\Sigma_{2}\succ\cdots \succ\Sigma_{l}$ and $x\in X$, let $\theta_{\succsim}(x):=(x_{1},x_{2},\cdots,x_{l})$, where $x_{k}:=|\Sigma_{k}[x]|$, that is, the number of elements in $\Sigma_{k}$ that contains $x$. Furthermore, let
$T_{k} := \bigcap  ( \bigcup_{k' \leq k} \Sigma_{k'} ) = \{ x \in X \mid \forall k' \leq k,\forall S \in \Sigma_{k'},x \in S \}$; let 
\begin{equation}
e_{\succsim}(x):=
\begin{cases}
0 &\text{if }x\notin T_{1},\\
\max\{k\in[l]\mid x\in T_{k}\} &\text{otherwise.}
\end{cases}
\end{equation}
For two finite vectors $a:=(a_{1},\cdots,a_{l})$ and $b:=(b_{1},\cdots ,b_{l})$, we denote $a\geq^{L}b$ if and only if $a_{k}=b_{k} \text{ for all } k\in[l]$ or $\exists \hat{k}\in[l]$  such that $x_{k}=y_{k}$ for all $k<\hat{k}$ and $x_{\hat{k}}>y_{\hat{k}}.$ Furthermore, we denote $a\geq^{L}b$ if and only if $a\geq^{L}b$ and $\neg(b\geq^{L}a)$.

\begin{definition}\label{def:threesrs}
We define SRSs $R^L,R^P,R^{IIS}$ as follows. For any $\succsim\in\domain$ with $\succsim:\Sigma_{1}\succ\cdots\succ\Sigma_{l}$ and $x,y\in X$, let
\begin{align*}
&xR^{L}_{\succsim}y \iff \theta_{\succsim}(x)\geq^{L}\theta_{\succsim}(y), \\
&xR^{P}_{\succsim}y \iff x_{1}\geq y_{1},\\
&xR^{IIS}_{\succsim}y \iff e_{\succsim}(x)\geq e_{\succsim}(y).
\end{align*}
\end{definition}

\begin{example}
Let $X=\{1,2,3\}$ and $\succsim\in\domain$ such that:
\[
\Sigma_1 =\{\{1\}, \{3\},\{1,2\}\} \succ \Sigma_2 =\{ \{2\}\} \succ \Sigma_3=\{\{1,3\}, \{2,3\}, \{1,2,3\}, \emptyset\}. 
\]
We have that $\theta_{\succsim}(1)=(2,0,2)$, $\theta_{\succsim}(2)=(1,1,2)$ and $\theta_{\succsim}(3)=(1,0,3)$, while $e_{\succsim}(x)=0$ doe each $x \in X$. So, the three SRSs of Definition \ref{def:threesrs} yields the following weak orders on $X$:
\begin{itemize}
    \item
$1\ P^{L}_{\succsim}\ 2\ P^{L}_{\succsim}\ 3,$
\item $1\ P^{P}_{\succsim}\ 2 \ I^{P}_{\succsim}\ 3$, 
\item $
1 \ I^{IIS}_{\succsim}\ 2 \ I^{IIS}_{\succsim}\ 3$.
\end{itemize}
\end{example}

As shown in the previous example, $R^{L}$, $R^{P}$ and $R^{IIS}$ are well distinct SRSs. Nevertheless, for any $\succsim\in\domain$ and $x,y \in X$, it is easy to check from the definition of SRSs $R^{L}$, $R^{P}$ and $R^{IIS}$ that the following logical implications hold:

\begin{proposition}
For any $\succsim\in\domain$ and $x,y\in X$, we have:
\begin{itemize}
    \item[i)] if $x I^{L}_{\succsim} y$, then $x I^{P}_{\succsim} y$ and $x I^{IIS}_{\succsim} y$;
    \item[ii)] if $x P^{L}_{\succsim} y$, then $x R^{P}_{\succsim} y$ and $x R^{IIS}_{\succsim} y$;
        \item[iii)] if $x P^{P}_{\succsim}y$ , then $x P^{L}_{\succsim} y$ and $x R^{IIS}_{\succsim} y$;
        \item[iv)] if $x P^{IIS}_{\succsim} y$, then $x P^{L}_{\succsim}\ y$ and $x R^{P}_{\succsim} y$.
\end{itemize}
\end{proposition}
\begin{proof}
Implication $(i)$ follows directly from the fact that if $x I^{L}_{\succsim} y$ then, by definition of $R^{L}$, $\theta_{\succsim}(x)=\theta_{\succsim}(y)$.   

To prove implication $(iii)$, consider $x,y \in X$ such that $x P^{P}_{\succsim} y$. Then $x_1 >y_1$, which directly implies $\theta_{\succsim}(x)>^{L}\theta_{\succsim}(y)$  and, by definition, $x P^{L}_{\succsim} y$. Now, let $\Sigma_1$ be the best equivalence class in $\succsim$. To prove that $x R^{IIS}_{\succsim} y$, we distinguish two cases: either $|\Sigma_1|=x_1>y_1$, where it necessarily holds that  $e_{\succsim}(x)\geq 1$ and $e_{\succsim}(y) =0$ (so, $x P^{IIS}_{\succsim} y$), or $|\Sigma_1|>x_1>y_1$ where it necessarily holds that  $e_{\succsim}(x)=e_{\succsim}(y)= 0$ (so, $x I^{IIS}_{\succsim} y$).

To prove implication $(iv)$, consider $x,y \in X$ such that $x P^{IIS}_{\succsim} y$. Then, $e_{\succsim}(x)>e_{\succsim}(y)= 0$, which directly implies $\theta_{\succsim}(x)>^{L}\theta_{\succsim}(y)$ and $x P^{L}_{\succsim} y$, by definition. To prove that $x R^{P}_{\succsim} y$, we distinguish two cases: either  $e_{\succsim}(x)= 1$ and $e_{\succsim}(y) =0$, implying $x_1=|\Sigma_1|>y_1$ (so, $x P^{P}_{\succsim} y$), or $e_{\succsim}(x)> e_{\succsim}(y)\geq 1$, implying $x_1=y_1$ (so, $x I^{P}_{\succsim} y$).

Finally, implication $(ii)$ follows from implications $(iii)$ and $(iv)$.
\end{proof}

\section{Results}
\label{section:results}

\subsection{Main theorems}
\label{subsection:main theorems}
In the paper \cite{Suzuki2025} $R^{IIS}$ is characterized using the following five axioms: 

\begin{theorem}
\label{theorem:IIS}
(\cite{Suzuki2025}) An SRS $R:\domain\rightarrow\codomain$ satisfies NT, WIVIP, IWS, IBS, and INUI if and only if it is $R^{IIS}$.
\end{theorem}

Our main result characterizes lex-cel $R^{L}$ and plurality $R^{P}$ by substituting the independence axioms in Theorem \ref{theorem:IIS}.

\begin{theorem}
\label{theorem:lex-cel}
An SRS $R:\domain\rightarrow\codomain$ satisfies WIVIP, IWS, and SSI if and only if it is $R^L$.    
\end{theorem}
\begin{proof}
The 'if' part is straightforward. We now prove the 'only if' part. Let $R$ be an SRS satisfying the axioms stated. We prove the following. 
(i) for any $\succsim\in\domain$ and $x,y\in X$, if $\theta_{\succsim}(x)>^{L}\theta_{\succsim}(y)$, then $xP_{\succsim}y$; (ii) for any $\succsim\in\domain$ and $x,y\in X$, if $\theta_{\succsim}(x)=\theta_{\succsim}(y)$, then $xI_{\succsim}y$. 

Proof of (i): Let $\succsim\in\domain$ with $\succsim:\Sigma_{1}\succ\cdots\succ\Sigma_{l}$. Let $x,y\in X$, $\theta_\succsim(x):=(x_1,x_2,\cdots,x_l)$ and $\theta_\succsim(y):=(y_1,y_2,\cdots,y_l)$. We assume $\theta_\succsim(x)>^L\theta_\succsim(y)$. By definition, there exists $\hat{k}\in[l]$ such that $x_k=y_k$ for all $k<\hat{k}$ and $x_{\hat{k}}>y_{\hat{k}}$. Let $\Gamma\subseteq\Sigma_{\hat{k}}[x]\setminus\Sigma_{\hat{k}}[y]$ such that $|\Gamma|=x_{\hat{k}}-y_{\hat{k}}$. Let 
\begin{align*}
\succsim_{0}&:\Gamma\succ_{0}(\Sigma_{1}\cup\cdots\cup\Sigma_{l})\setminus\Gamma\\
\succsim_{1}&:\Sigma_{1}\succ_{1}\Gamma\succ_{1}(\Sigma_{2}\cup\cdots\cup\Sigma_{l})\setminus\Gamma\\
\succsim_{2}&:\Sigma_{1}\succ_{2}\Sigma_{2}\succ_{2}\Gamma\succ_{2}(\Sigma_{3}\cup\cdots\cup\Sigma_{l})\setminus\Gamma\\
&\vdots\\
\succsim_{\hat{k}-1}&:\Sigma_{1}\succ_{\hat{k}-1}\Sigma_{2}\succ_{\hat{k}-1}\cdots\succ_{\hat{k}-1}\Sigma_{\hat{k}-1}\succ_{\hat{k}-1}\Gamma\succ_{\hat{k}-1}(\Sigma_{\hat{k}}\cup\cdots\cup\Sigma_{l})\setminus\Gamma\\
\succsim_{\hat{k}}&:\Sigma_{1}\succ_{\hat{k}}\cdots\succ_{\hat{k}}\Sigma_{\hat{k}-1}\succ_{\hat{k}}\Sigma_{\hat{k}}\succ_{\hat{k}}(\Sigma_{\hat{k}+1}\cup\cdots\cup\Sigma_{l}).
\end{align*}
Thus, $\succsim_{0}$ ranks all elements indifferently, except for the elements in $\Gamma$ that are placed in the best equivalence class. For each $k=1,2,\cdots,\hat{k}-1$, $\succsim_{k+1}$ is obtained from $\succsim_{k}$ by shifting $\Sigma_{k}\setminus\Gamma$ above $\Gamma$. Finally, $\succsim_{\hat{k}}$ is obtained from $\succsim_{\hat{k}-1}$ by shifting $\Sigma_{\hat{k}}\setminus\Gamma$ to the class $\Gamma$. 

Subsequently, WIVIP implies $xP_{\succsim_{0}}y$. From the definitions of $\hat{k}$ and $\Gamma$, we can infer that $|\Sigma_{k}[x]|=|\Sigma_{k}[y]|$ for all $k\leq \hat{k}$. Hence, SSI requires that $R_{\succsim}\mid_{\{x,y\}}=R_{\succsim'}\mid_{\{x,y\}}$ for all $k\leq \hat{k}$. Therefore, we conclude that $xP_{\succsim_{\hat{k}}}y$. Notably, $\succsim$ is obtained from $\succsim_{\hat{k}}$ by partitioning the worst equivalence class. 
Hence, a solution satisfying axiom IWS yields the relation. 

Proof of (ii): Assume that there exist $\succsim\in\domain$ and $x,y\in X$ such that $\theta_{\succsim}(x)=\theta_{\succsim}(y)$ and $\neg(xI_{\succsim}y)$. As $R_{\succsim}$ is assumed to be complete ($\because$ $R_{\succsim}\in\codomain$), we obtain either $xP_{\succsim}y$ or $yP_{\succsim}x$. Without a loss of generality, we can assume that $xP_{\succsim}y$. Let
\begin{align*}
(\succsim=)\succsim_{1}&:\Sigma_{1}\succ\cdots\succ\Sigma_{l},\\
\succsim_{2}&:\Sigma_{1}\cup\Sigma_{2}\succ\Sigma_{3}\succ\cdots\succ\Sigma_{l},\\
&\vdots\\
\succsim_{l}&:\Sigma_{1}\cup\cdots\cup\Sigma_{l}.
\end{align*}

From SSI, we obtain $R_{\succsim_{k}}\mid_{\{x,y\}}=R_{\succsim_{k+1}}\mid_{\{x,y\}}$ for all $k\in[l-1]$. Together with the assumption $xP_{\succsim}y$, we obtain $xP_{l}y$. Let $\succsim':\left(\Sigma_{1}\cup\cdots\cup\Sigma_{l}\right)\setminus\{\{x\},\{y\}\}\succ\{\{x\},\{y\}\}$. Thus, $\succsim'$ is obtained from $\succsim_{l}$ by reducing the positions of $\{x\}$ and $\{y\}$. Using $xP_{\succsim_{l}}y$ and SSI, we conclude that $xP_{\succsim'}y$. 

Finally, let $\succsim'':\left(\Sigma_{1}\cup\cdots\cup\Sigma_{l}\right)\setminus\{\{x\},\{y\}\}\succ\{\{y\}\}\succ\{\{x\}\}.$ Thus, $\succsim''$ is obtained from $\succsim'$ by resolving ties in the worst equivalence class. Consequently, $xP_{\succsim'}y$, combined with IWS, implies that $xP_{\succsim}y$. However, we confirm that $\theta_{\succsim''}(y)>^{L}\theta_{\succsim''}(x)$. Hence, Case (i) states that $yP_{\succsim}x$, leading to a contradiction. 
\end{proof}

\begin{theorem}
\label{theorem:plurality}
An SRS $R:\domain\rightarrow\codomain$ satisfies NT, WIVIP, TO, and SI if and only if it is $R^{P}$. 
\end{theorem}

\begin{proof}
The 'if' part is straightforward. We prove the “only if” aspect. Let $R$ be an SRS satisfying the axioms stated. Let $\succsim\in\domain$ with $\succsim:\Sigma_{1}\succ\cdots\succ\Sigma_l$ and $x,y\in X$. Let $\theta_\succsim(x):=(x_1,x_2,\cdots,x_l)$ and $\theta_\succsim(y):=(y_1,y_2,\cdots,y_l)$. We prove that (i) if $x_1=y_1$, then $xI_\succsim y$, and (ii) if $x_1>y_1$, then $xP_\succsim y$. 

Proof of (i): Assume that $x_1=y_1$. Define
\begin{align*}
\succsim_{1}(:=\succsim)&:\Sigma_{1}\succ\cdots\succ\Sigma_{l},\\
\succsim_{2}&:\Sigma_{1}\succ\Sigma_{2}\cup\cdots\cup\Sigma_{l}.\\
\end{align*}

From TO, we obtain $R_{\succsim_{1}}=R_{\succsim_{2}}$. 

We define $S\in\mathfrak{X}$ and $\succsim_{3}\in\domain$ as follows: 
\begin{itemize}
    \item If $\Sigma_{1}[x]\cap\Sigma_{1}[y]\neq\emptyset$, then let $S\in\Sigma_{1}[x]\cap\Sigma_{1}[y]$ and $\succsim_{3}:=\succsim_{2}$. 
    \item Otherwise, i.e., if $\Sigma_{1}[x]\cap\Sigma_{1}[y]=\emptyset$, then let $S:=\{x,y\}$ and $\succsim_{3}:\Sigma_{1}\cup\{S\}\succ\left(\Sigma_{2}\cup\cdots\cup\Sigma_{l}\right)\setminus\{S\}.$ 
\end{itemize}

In either case, we obtain 
$R_{\succsim_{2}}\mid_{\{x,y\}}=R_{\succsim_{3}}\mid_{\{x,y\}}$ ($\because$ for the former case, it is trivial because $\succsim_{2}=\succsim_{3}$; for the latter case, $\Sigma_{1}[x]\cap\Sigma_{1}[y]=\emptyset$ implies that $\{x,y\},\{x,y,z\}\in\Sigma_{2}\cup\cdots\cup\Sigma_{l}$. Thus, $\{S\}\subsetneq\Sigma_{2}\cup\cdots\cup\Sigma_{l}$. Hence, the equation holds for SI). 

 Now, let
$\succsim_{4}:\{S\}\succ\left(\Sigma_{1}\cup\cdots\cup\Sigma_{l}\right)\setminus\{S\}$. Thus, $\succsim_{4}$ is obtained from $\succsim_{3}$ by reducing the position of $\Sigma_{1}\setminus\{S\}$. By the assumption $x_{1}=y_{1}$, it holds $|\left(\Sigma_{1}\setminus\{S\}\right)[x]|=|\left(\Sigma_{1}\setminus\{S\}\right)[y]|$. Thus, SI implies that $R_{\succsim_{3}}\mid_{\{x,y\}}=R_{\succsim_{4}}\mid_{\{x,y\}}$. In conclusion, the above argument shows that $R_{\succsim}\mid_{\{x,y\}}=R_{\succsim_{4}}\mid_{\{x,y\}}$. NT implies that $xI_{\succsim_{4}}y$. Thus, we obtain $xI_{\succsim}y$.


Proof of (ii): Assume that $x_1>y_1$. From the definitions of $x_1$ and $y_1$, there exists $\Gamma\subseteq\Sigma_1[x]\setminus\Sigma_1[y]$ such that $|\Gamma|=x_1-y_1$. Subsequently, it follows that $|\left(\Sigma_{1}\setminus\Gamma\right)[x]|=|\left(\Sigma_{1}\setminus\Gamma\right)[y]|$. Let
\begin{align*}
\succsim_{1}&:\Gamma\succ\left(\Sigma_{1}\cup\cdots\cup\Sigma_{l}\right)\setminus\Gamma,\\
\succsim_{2}&:\Sigma_{1}\succ\Sigma_{2}\cup\cdots\cup\Sigma_{l}.
\end{align*}
Using WIVIP, we obtain $xP_{\succsim_{1}}y$. Because $\succsim_{2}$ is obtained from $\succsim_{1}$ by shifting the position of $\Sigma_{1}\setminus\Gamma$, the SI requires $R_{\succsim_{1}}\mid_{\{x,y\}}=R_{\succsim_{2}}\mid_{\{x,y\}}$. This implies that $xP_{\succsim_{2}}y$. As the best equivalence class of $\succsim_{2}$ is the same as that of $\succsim$, TO implies that $R_{\succsim_{2}}=R_{\succsim}$. Therefore, we conclude that $xP_{\succsim}y$. 
\end{proof}

Notably, $F^{L}$ satisfies NT trivially. Hence, the three characterizations (Theorems \ref{theorem:IIS}, \ref{theorem:lex-cel}, and \ref{theorem:plurality}) state that (i) the three SRSs share NT and WIVIP, and (ii) they differ only in the independence axioms (IWS, IBS, INUI, TO, SI, and SSI). Thus, the theorems effectively capture the differences among the three SRSs in terms of independence. 

By dividing the axioms into blocks, we can achieve a better understanding of the characteristics between the three SRSs. 

\begin{align*}
    R^{IIS} &= \text{[NT+WIVIP]}
    + \text{[IWS+IBS]} +\text{[INUI]}, \\
    R^{L}   &= \text{[NT+WIVIP]}
    + \text{[IWS]} +\text{[SSI}(\Rightarrow
    \text{SI)]},\\
    R^{P} &=\text{[NT+WIVIP]}+\text{[TO}(
    \Rightarrow \text{IWS)]}+\text{[SI]}\\
\end{align*}

The first block [NT+WIVIP] shows the common basis of the three SRSs. 

The axioms in the second block (IWS, IBS, and TO) concern independence from the extremal elements (see Remark \ref{remark:independence from the ends}). Notably, TO implies IWS. Thus, from the perspective of independence from the extremal elements, we prove that (i) each $R^{IIS}$, $R^{L}$, and $R^{P}$ satisfies IWS as a common basis; (ii) $R^{IIS}$ also satisfies IBS; 
and (iii) $R^{P}$ satisfies TO (i.e., independent of elements worse than the best equivalence class). In Subsection \ref{subsection:independence axioms and SRSs: a review}, we show that $R^{IIS}$ does not satisfy TO, $R^{L}$ does not satisfy IBS nor TO, and $R^{P}$ does not satisfy IBS. Hence, the second block effectively describes the differences between the three SRSs.

The axioms in the third block (SI, SSI, and INUI) are related to independence from shifts of certain family of sets $\Delta$ (see the discussion immediately following Definition \ref{definition:INUI}). In Subsection \ref{subsection:independence axioms and SRSs: a review}, we confirm that $R^{IIS}$ does not satisfy SI or SSI, $R^{L}$ does not satisfy INUI, and $R^{P}$ does not satisfy SSI or INUI. Hence, the third block effectively captures the difference between the three SRSs in terms of their independence from the slide of $\Delta$. 

\subsection{Independence axioms and SRSs: a review}
\label{subsection:independence axioms and SRSs: a review}

 We add additional SRSs for further discussion. For two finite vectors $a:=(a_{1},\cdots,a_{l})$ and $b:=(b_{1},\cdots ,b_{l})$, we denote $a\geq^{DL}b $ if and only if $ a_{k}=b_{k}$, for all $k\in[l]$, or $\exists \hat{k}\in[l]$ such that $x_{k}=y_{k}$,  for all $k>\hat{k}$ and $x_{\hat{k}}<y_{\hat{k}}$.
We define the \textit{dual lexicographic excellence solution} (dual lex-cel) $R^{DL}$ for any $\succsim\in\domain$ and $x,y\in X$, $xR^{DL}_{\succsim}y\iff \theta_{\succsim}(x)\geq^{DL}\theta_{\succsim}(y)$. 

We define the \textit{constant-$X$ rule} $R^{\text{constX}}$ as  $xR^{\text{constX}}_{\succsim}y$, for all $x,y\in X$ and $\succsim\in\domain$. 

We define \textit{split-plurality} $R^{splitP}$, for any $\succsim\in\domain$ and for any $x,y\in X$, $xR^{\text{splitP}}_{\succsim}y\iff s^{\text{split}}_{\succsim}(x)\geq s^{\text{split}}_{\succsim}(y)$, where $s^{\text{split}}_{\succsim}(z):=\sum_{S\in\Sigma_{1}:z\in S} \frac{1}{|S|}$. 

Let $\vartriangleright$ be the fixed linear order of $X$. We define plurality with tie-breaking by $\vartriangleright$, denoted as $R^{P,\vartriangleright}$, for any $\succsim\in\domain$ and $x,y\in X$, let $xR^{
P,\vartriangleright}_{\succsim}(y)\iff (xP_{\succsim}^{P}y\ \text{or}\ (xI^{P}_{\succsim}y\text{ and }x\vartriangleright y))$.

Table \ref{table:SRSs and independence axioms} provides an overview of the relationship between SRSs and the axioms. In each cell, 1 (0) indicates that the SRS in the corresponding row satisfies (does not satisfy) the axiom in the column. For instance, all seven SRSs except $R^{DL}$ satisfy IWS. 
\begin{table}[h]
\centering
    \caption{SRSs and independence axioms}
    \label{table:SRSs and independence axioms}

\begin{tabular}{|c||c|c|c|c|c|c|c|c|c|}
\hline
 & NT & WIVIP & IWS & IBS & TO & SI & SSI & INUI & IIC \\\hline\hline
$R^{IIS}$ & 1 & 1 & 1 & 1 & 0 & 0 & 0 & 1 & 0 \\
\hline
$R^{L}$ & 1 & 1 & 1 & 0 & 0 & 1 & 1 & 0 & 1 \\
\hline
$R^{P}$ & 1 & 1 & 1 & 0 & 1 & 1 & 0 & 0 & 0 \\
\hline
$R^{DL}$ & 1 & 1 & 0 & 1 & 0 & 1 & 1 & 0 & 1 \\
\hline
$R^{\text{constX}}$ & 1 & 0 & 1 & 1 & 1 & 1 & 1 & 1 & 1 \\
\hline
$R^{\text{splitP}}$ & 1 & 1 & 1 & 0 & 1 & 0 & 0 & 0 & 0 \\
\hline
$R^{P,\vartriangleright}$ & 0 & 1 & 1 & 0 & 1 & 1 & 0 & 0 & 0 \\
\hline
\end{tabular}
\end{table}

Most cells are verified in a straightforward manner. Therefore, we present a sketch of the proofs for some nontrivial cells as follows: 

\begin{itemize}
\item On the row of $R^{IIS}$. See \cite{Suzuki2025} for NT, WIVIP, IWS, IBS, and INUI. The cell of TO is straightforward. $R^{IIS}$ does not satisfy SI, SSI, and IIC, indeed, let $x,y\in X$ and $\succsim:\{\{x\}\}\succ\{\{x\}\}^c$ and take $\succsim':\{\{x\},X\}\succ\{\{x\},X\}^c$ (i.e., a weak order obtained from $\succsim$ by shifting the position of $X$). Subsequently, each SI, SSI, and IIC requires that the social ranking between $x$ and $y$ remains the same between $\succsim$ and $\succsim'$. However, we obtain $xP^{IIS}_{\succsim}y$ and $xI^{IIS}_{\succsim'}y$. This proves that $R^{IIS}$ satisfies none of SI, SSI, and IIC. 
\item On the row of $R^{L}$ (the row of $R^{DL}$ is also verified similarly). See \cite{Bernardi2019} for NT, IWS, and IBS. By definition of $\geq^{L}$, the cells of TO, SI, SSI, and IIC are straightforward. $R^{L}$ does not satisfy INUI, indeed, let $x,y,z\in X$, and let $\succsim,\succsim'\in\domain$ be such that $\succsim:\mathfrak{X}$ and $\succsim':\Sigma\succ\mathfrak{X}\setminus\Sigma$, where $\Sigma:=\{\{x\},\{z\}\}$. Since $x,y\notin \bigcap \Sigma$, INUI demands that the social ranking between $x$ and $y$ remains the same between $\succsim$ and $\succsim'$. However, we obtain $xI^{L}_{\succsim}y$ and $xP_{\succsim'}y$. This proves that $R^{L}$ does not satisfy INUI. 
\item The rows of $R^{P}$, $R^{\text{constX}}$, and $R^{P,\vartriangleright}$ are straightforward. 
\item On the row of $R^{\text{splitP}}$. The cells for NT, WIVIP, IWS, IBS, and TO are straightforward. To indicate that $R^{\text{splitP}}$ does not satisfy SI or SSI, let $\succsim:\{X\}\succ\{X\}^c$, and $\succsim':\{X,\{x\},\{y,z\}\}\succ\{X,\{x\},\{y,z\}\}^c$. Then, SI and SSI demand that the social ranking between $x$ and $y$ remains the same between $\succsim$ and $\succsim'$. However, $xI^{\text{splitP}}_{\succsim}y$ and $xP^{\text{splitP}}_{\succsim'}y$. Hence, $R^{\text{splitP}}$ does not satisfy SI or SSI. One can prove the cell of INUI using the same approach as for the row $R^{L}$. To prove that $R^{\text{splitP}}$ does not satisfy IIC, let $\succsim:\{\{x\}\}\succ\{\{x\}\}^c$ and $\succsim':\{X\}\succ\{\{x\}\}\succ\{X,\{x\}\}^c$. Subsequently, IIC demands that the social ranking between $x$ and $y$ remains the same between $\succsim$ and $\succsim'$. However, we obtain $xP_{\succsim}^{\text{splitP}}y$ and $xI^{\text{splitP}}_{\succsim'}y$. Hence, $R^{\text{splitP}}$ does not satisfy IIC. 
\end{itemize}

As shown in Table \ref{table:SRSs and independence axioms}, we obtain the following. 

\begin{proposition}
The axioms in Theorems \ref{theorem:lex-cel} and \ref{theorem:plurality} are logically independent.  \end{proposition}

\section{Conclusion}
\label{section:conclusion}

This study examines SRSs in terms of independence axioms. We provide new characterizations of lex-cel $R^{L}$ (Theorem \ref{theorem:lex-cel}) and plurality $R^{P}$ (Theorem \ref{theorem:plurality}) by substituting the independence axioms in the existing characterization of IIS $R^{IIS}$. Our results reveal the characteristics of these three SRSs in terms of alternative combinations of independence axioms and shed new insight on their fundamental connections. 

\section*{Declarations}
\subsection*{Funding}
Takahiro Suzuki was supported by JSPS KAKENHI Grant
Number JP21K14222 and JP23K22831. Stefano Moretti acknowledges financial support from the ANR project THEMIS (ANR-20-CE23-0018). Michele Aleandri is member of GNAMPA of the Istituto Nazionale di Alta Matematica (INdAM) and is supported by Project of Significant National Interest – PRIN 2022 of title “Impact of the Human Activities on the Environment and Economic Decision Making in a Heterogeneous Setting: Mathematical Models and Policy Implications”- Codice Cineca: 20223PNJ8K- CUP I53D23004320008.
\subsection*{Competing interests}
The authors have no relevant financial or non-financial interests to disclose.


\end{document}